\documentclass[letterpaper,11pt]{article}
\usepackage{fullpage}
\usepackage{mathtools,amsmath,amsthm, amssymb}
\usepackage{paralist}
\usepackage{algpseudocode,algorithm}

\newtheorem{thm}{Theorem}[section]

\newtheorem{lem}[thm]{Lemma}

\newtheorem{Claim}{Claim}
\theoremstyle{definition}

\theoremstyle{remark}

\newcommand{\R}{\mathbb R}

\newcommand{\F}{\mathbb F}
\newcommand{\E}{\mathbb E}
\newcommand{\OPT}{\mathrm{OPT}}

\newcommand{\ALG}{\mathrm{ALG}}

\newcommand{\I}{\mathcal{I}}

\newcommand{\M}{\mathcal M}
\newcommand{\PP}{\mathcal P}

\newcommand{\D}{\mathcal D}

\DeclareMathOperator*{\spa}{span}

\DeclareMathOperator*{\Bin}{Bin}

\begin{document}
\title{Matroid Secretary Problem in the Random Assignment Model\footnote{This work was partially supported by NSF
  contract CCF-0829878 and by ONR grant N00014-05-1-0148.}}
\author{Jos\'e A. Soto\thanks{MIT,
Dept.~of Math., Cambridge, MA 02139. \texttt{jsoto@math.mit.edu}.}}
\date{}

\maketitle
\begin{abstract}
In the \emph{Matroid Secretary Problem}, introduced by Babaioff et al.~[SODA 2007], the elements of a given matroid are presented to an online algorithm in random order. When an element is revealed, the algorithm learns its weight and decides whether or not to select it under the restriction that the selected elements form an independent set in the matroid. The objective is to maximize the total weight of the chosen elements. In the most studied version of this problem, the algorithm has no information about the weights beforehand. We refer to this as the \emph{zero information model}. In this paper we study a different model, also proposed by Babaioff et al., in which the relative order of the weights is random in the matroid. To be precise, in the \emph{random assignment model}, an adversary selects a collection of weights that are randomly assigned to the elements of the matroid. Later, the elements are revealed to the algorithm in a random order independent of the assignment.

Our main result is the first constant competitive algorithm for the matroid secretary problem in the random assignment model. This solves an open question of Babaioff et al. Our algorithm achieves a competitive ratio of $2e^2/(e-1)$. It exploits the notion of \emph{principal partition} of a matroid, its decomposition into \emph{uniformly dense minors}, and a $2e$-competitive algorithm for uniformly dense matroids we also develop. As additional results, we present simple constant competitive algorithms in the zero information model for various classes of matroids including cographic, low density and the case when every element is in a small cocircuit. In the same model, we also give a $ke$-competitive algorithm for $k$-column sparse linear matroids, and a new $O(\log r)$-competitive algorithm for general matroids of rank $r$ which only uses the relative order of the weights seen and not their numerical value, as previously needed.
\end{abstract}
\thispagestyle{empty}

\newpage
\setcounter{page}{1}
\section{Introduction}
In the simplest form of the \emph{secretary problem}, an employer wants to select the best secretary among $n$ candidates arriving in random order. Once a secretary is interviewed, the employer must decide immediately whether to accept the candidate or not and that decision is final. Lindley \cite{lindley_dynamic_1961} and Dynkin \cite{dynkin_optimum_1963} have shown that sampling the first $\lfloor n/e\rfloor$ candidates and then selecting the first one whose value is higher than all the sampled ones gives a probability of at least $1/e$ of selecting the best secretary and that no algorithm can beat this constant. An important generalization of this problem with many applications is known as the \emph{multiple choice secretary problem} (see \cite{kleinberg_multiple-choice_2005}). In this problem we wish to select a set of at most~$k$ secretaries from a pool of $n$ applicants and the objective is to select a group of combined value as high as possible.

Babaioff et al.~\cite{babaioff_online_2008} introduce the \emph{generalized secretary problem} as a natural class of extensions of the previous problem in which the set returned by the algorithm must obey some combinatorial restriction. In this setting, a finite set $E$ with hidden nonnegative weights and a collection of subsets $\I \subseteq 2^E$ closed under inclusion are given. The collection $\I$ describes the sets of elements that can be simultaneously accepted (these are the \emph{feasible sets} or the \emph{domain} of the problem). The elements of $E$ are presented to an online algorithm in random order. When an element is revealed, the algorithm learns its weight and decides whether or not to accept it under the restriction that the set of accepted elements is feasible. This decision is irreversible and it must be taken before the next element is revealed. The objective is to output a feasible set of maximum total weight.

We remark that other lines of generalizations of the multiple choice secretary problem having different objective functions have also been considered. These generalizations include, among others, minimizing the sum of the relative ranks of the selected elements (studied by Ajtai et al.~\cite{ajtai_1995}), the weighted and time discounted secretary problems of Babaioff et al.~\cite{babaioff_secretary_2009}, the $J$-choice $K$-best secretary problem studied by Buchbinder et al.~\cite{buchbinder2010secretary} and the submodular secretary problem of Bateni et al.~\cite{bateni-submodular2010}.


The generalized secretary problem is of interest due to its connection to online auctions. In both the original and multiple choice secretary problems, we can regard the algorithm as an auctioneer having one or many identical items, and the secretaries as agents arriving at random times, each one having a different valuation for the item. The goal of the algorithm is to assign the items to the agents as they arrive while maximizing the total social welfare. In more complex situations, the algorithm may be considered to have access to a collection of goods that it wishes to assign to agents, subject to some restrictions. In many cases, these restrictions can be modeled by matroid constraints. For that reason, the \emph{matroid secretary problem}, in which the feasible sets are the independent sets of a matroid is of special interest (see e.g. a survey of Babaioff et al.~\cite{babaioff_online_2008}).

Notice that the difficulty of the problem changes depending on the information we know beforehand about the weights. We recognize four different models in increasing order of difficulty. \begin{compactitem}
  \item \textbf{Full information model:} The weights are chosen i.i.d. from a known distribution.
  \item \textbf{Partial information model:} The weights are chosen i.i.d. from an unknown distribution.
  \item \textbf{Random assignment model:} An adversary chooses a list of nonnegative weights, which are then assigned to the elements using a uniform random one-to-one correspondence, which is independent of the random order at which the elements are revealed.
  \item \textbf{Zero information model:} An adversary assigns the weights arbitrarily.
\end{compactitem}

The difficulty may also change depending on whether the algorithm learns the actual weight of the elements or just the relative order of the weights seen so far. See the surveys of Freeman~\cite{freeman1983secretary} and Ferguson \cite{ferguson_who_1989} for references and variations of these models in the classical secretary problem. Note that for both the classical problem and the multiple choice problem, the random assignment and the zero information models coincide.

There has been a significant amount of work on the matroid secretary problem under zero information. Constant competitive algorithms are known for partition matroids (this corresponds to the classical~\cite{lindley_dynamic_1961,dynkin_optimum_1963} and multiple choice secretary problem~\cite{kleinberg_multiple-choice_2005,babaioff_knapsack_2007}) and also for graphic and transversal matroids~\cite{babaioff_matroids_2007,dimitrov_competitive_2008,korula2009algorithms,babaioff_secretary_2009}. It is also known~\cite{babaioff_matroids_2007} that if a matroid admits a constant competitive algorithm under zero information, then so do its restrictions and truncations. For general matroids, the best algorithm known so far, due to Babaioff et al.~\cite{babaioff_matroids_2007}, is $O(\log r)$-competitive, where $r$ is the rank of the matroid.

Non-matroidal domains have also been considered in the literature. Babaioff et al.~\cite{babaioff_knapsack_2007} show a $10e$-competitive algorithms for knapsack domains even in the case where both the weights and lengths are revealed online. Korula and Pal~\cite{korula2009algorithms} give constant competitive algorithms for some cases of intersection of partition matroids under zero information, specifically for matchings in graphs and hypergraphs where the edges have constant size.

Not every domain admits constant competitive algorithms. Babaioff et al. \cite{babaioff_matroids_2007} have shown a particular domain for which no algorithm has a competitive ratio smaller than $o(\log n/\log\log n)$ even in the full information model. However, matroid domains have the following special property: If we are allowed to reject elements which have been previously accepted, while keeping at every moment an independent set, then it is possible to output the optimum independent set no matter in which order the elements are presented. This intuition motivated Babaioff et al.~\cite{babaioff_matroids_2007, babaioff_online_2008} to conjecture that the matroid secretary problem admits a constant competitive algorithm even under zero information. According to these authors, this question is also non-trivial for the random assignment model, and for the model in which the \emph{order} the elements are presented is adversarial, but the weights are randomly assigned from a pool of hidden values.

\begin{paragraph}{Main Result}
In this paper, we answer the last question affirmatively for the random assignment model, exhibiting a $2e^2/(e-1)$-competitive algorithm for any matroid domain. We remark here that our results also apply to the partial and full information settings for which, as far as we know, no previous results existed. Our algorithm uses as building block an algorithm for uniformly dense matroids we also develop.

Uniformly dense matroids are matroids for which the \emph{density} of a set, that is, the ratio of its cardinality to its rank, is at most the density of the entire ground set. The simplest examples of these are precisely the uniform matroids. Uniformly dense matroids and uniform matroids of the same rank over the same ground set behave similarly, in the sense that the distribution of the rank of a random set is similar for both matroids. We use this fact to devise a $2e$-competitive algorithm for these matroids in the random assignment model. In order to extend this algorithm to general matroids we exploit some notions coming from the \emph{theory of principal partitions} of a matroid, particularly its \emph{principal sequence}. Roughly speaking, the principal sequence of a matroid $\M$ is a decomposition of its ground set into a sequence of parts, each of which is the underlying set of a uniformly dense minor of $\M$. By employing independently the previous algorithm in each of these minors, we obtain an algorithm that returns an independent set of $\M$, while loosing only an extra factor of $1-1/e$ on its competitive ratio.
\end{paragraph}

\begin{paragraph}{Additional Results}
We also develop simple constant competitive algorithms under zero information for various classes of matroids including cographic, low density and the case when every element is in a small cocircuit. Also, we show a $ke$-competitive algorithm for the case when the matroid is representable by a matrix in which each column has at most $k$ non-zero entries. This result generalizes the $2e$-competitive algorithm for graphic matroids of Korula and Pal~\cite{korula2009algorithms}. Finally, we give a new $O(\log r)$-competitive algorithm for general matroids. Unlike the previous algorithm of Babaioff et al.~\cite{babaioff_matroids_2007}, our algorithm does not use the numerical value of the weights. It only needs the ability to make comparisons among seen elements.
\end{paragraph}
\section{Matroid Secretary Problem in the Random Assignment Model}\label{section2}

In this paper, we assume familiarity with basic concepts in matroid theory. For an introduction and specific results, we refer to Oxley's book \cite{oxley2006matroid}.

Consider a matroid $\M=(E,\I)$ with ground set $E=\{e_1,\ldots,e_n\}$. An adversary selects a set $W$ of $n$ nonnegative weights $w_1\geq \dots \geq w_n \geq 0$, which are assigned to the elements of the matroid via a random permutation $\sigma:[n]\to[n]$, i.e., the weight function of the elements $w:E\to W$ is such that $w(e_{\sigma(i)})=w_i$. The elements are then presented to an online algorithm via a random order $\pi:[n]\to[n]$. When an element is presented, the algorithm must decide whether to add it or not to the solution set under the condition that this set is independent in $\M$ at all times. The objective is to output a solution set $\ALG$ whose value $w(\ALG)=\sum_{e \in \ALG} w(e)$ is as high as possible.

We further assume that when the $i$-th element of the stream, $e_{\pi(i)}$, is presented, \emph{the algorithm only learns the relative order of the weight with respect to the previously seen ones}. This is, it can compare $w(e_{\pi(j)})$ with $w(e_{\pi(k)})$ for all $j,k\leq i$, but it can not use the actual numerical values of the weights. Without loss of generality, we can assume that there are no ties in $W$, because otherwise we can break them using a new random permutation $\tau$; if the comparison between two elements seen gives a tie, then we consider heavier the one having larger $\tau$-value.

For a given permutation $\sigma$, let $\OPT_\M(\sigma)$ be the the lexicographic first base of $\M$ according to the permutation $\sigma$. In other words, $\OPT_\M(\sigma)$ is the set obtained by applying the \emph{greedy procedure} that selects an element if it can be added to the previously selected ones while preserving independence in $\M$, on the sequence $e_{\sigma(1)}, e_{\sigma(2)}, \dots, e_{\sigma(n)}$. Standard matroid arguments imply that $\OPT_\M(\sigma)$ is a maximum independent set with respect to any weight function $v$ for which $v(e_{\sigma(1)})\geq \dots \geq v(e_{\sigma(n)})\geq 0.$ In particular, this is true for the weight function $w$ defined before. We will drop the subindex $\M$ in $\OPT_\M(\sigma)$ whenever there is no possible confusion.

We say that an online algorithm returning an independent set $\ALG$ is $\alpha$-competitive if for any selection of nonnegative weights $W$ given by the adversary, $\alpha\E_{\pi,\sigma}[w(\ALG)] \geq \E_{\sigma}[w(\OPT(\sigma))].$

As a particular case, consider the partial and full information models in which elements receive their weights independently from a fixed distribution $\D$ over the nonnegative reals. Since it is possible for the expected weight of the optimum and the expected weight of the set returned by the algorithm to be both infinite (for instance, if $\D$ has infinite mean), the concept of competitiveness for this scenario has to be slightly modified. We say that an algorithm returning an independent set $\ALG$ is $\alpha$-competitive in the partial or full information models if $\E_{\D,\pi}[\alpha w(\ALG) - w(\OPT)]\geq 0$.

We claim that any algorithm that is $\alpha$-competitive in the random assignment scenario is also $\alpha$-competitive in both the full and the partial information setting. To see this, consider an $\alpha$-competitive algorithm $\mathcal{A}$ under the former model, and apply to one of the latter. Note that the distribution of the maximum independent set is the same as the one obtained by first choosing a set $W$ of $n$ sample values from $\D$, and then assigning these values to the elements using a uniform random permutation~$\sigma$. For any realization of values of $W$ according to $\D$, algorithm $\mathcal{A}$ returns a set $\ALG$ such that $\E_{\pi,\sigma}[\alpha w(\ALG) - w(\OPT(\sigma))]\geq 0$. By taking expectation over the realizations of $W$ we prove the claim.

\section{Uniformly dense matroids}\label{section3}
Define the \emph{density} $\gamma(\M)$ of a \emph{loopless matroid}\footnote{A \emph{circuit} is a minimal non-independent set of a matroid. A \emph{loop} is an element $e$ such that $\{e\}$ is a circuit. A loopless matroid is a matroid having all singletons independent. } $\M=(E,\I)$ with rank function $r$, to be the maximum over all non-empty sets $|X|$, of the quantity $|X|/r(X)$.
The matroid $\M$ is \emph{uniformly dense} if $\gamma(\M)$ is attained by the entire ground set; that is, if $\frac{|X|}{r(X)}\leq \frac{|E|}{r(E)}$, for every non-empty $X \subseteq E$. Examples of uniformly dense matroids include uniform matroids, the graphic matroid of a complete graph, and all projective geometries $PG(r-1,\F)$. The following property of uniformly dense matroids is important for our analysis.
\begin{lem}\label{lem:randomgreedy}
Let $(x_1,\dots,x_j)$ be a sequence of different elements of a uniformly dense matroid chosen uniformly at random. The probability that element $x_j$ is selected by the greedy procedure on that sequence is at least $1-(j-1)/r$, where $r$ is the rank of the matroid.
\end{lem}
\begin{proof}
An element is selected by the greedy procedure only if it is outside the span of the previous elements. Denote by $A_i=\{x_1,\ldots,x_i\}$ to the set of the first $i$ elements of the sequence, and let $n$ be the number of elements of the matroid, then:
  \begin{align*}
    \Pr[x_j \text{ is selected}] &= \frac{n-\spa(A_{j-1})}{n-(j-1)} \geq \frac{n-r(A_{j-1})n/r}{n-(j-1)} \geq \frac{n-(j-1)n/r}{n-(j-1)} \geq 1-(j-1)/r,
  \end{align*}
  where the first inequality holds since the matroid is uniformly dense and the second holds because the rank of a set is always at most its cardinality.
\end{proof}

\subsection{A $2e$-competitive algorithm for uniformly dense matroids}

Consider the following algorithm for uniform matroids of rank $r$ due to Babaioff et al.~\cite{babaioff_knapsack_2007}: Maintain the set $T$ consisting on the $r$ heaviest elements seen so far (initialize this set with $r$ dummy elements that are considered to be lighter than everyone else). Observe the first $m=pn$ elements without adding them to the output; we refer to this set as the \emph{sample}. An element arriving later will be added to the output only if at the moment it is seen, it enters $T$ and the element that leaves $T$ is either in the sample or a dummy element. Babaioff et al.~have shown that this algorithm returns a set of at most $r$ elements and that by setting $p$ to be $1/e$, every element of the optimum is in the output of the algorithm with probability at least $1/e$, making this algorithm $e$-competitive for \emph{uniform} matroids even under zero information.

A slight modification of this algorithm is at least $2e$-competitive for uniformly dense matroids in the random assignment model. The full procedure is depicted in Algorithm \ref{algorithm}. The only differences with respect to the algorithm above are that \begin{inparaenum}[(i)]\item The number of elements sampled is given by a binomial distribution $\Bin(n,p)$ and \item Before adding an element to the output, we test if its addition maintains independence in the matroid\end{inparaenum}.

\begin{algorithm}[h!!!]
\caption{for uniformly dense matroids of $n$ elements and rank $r$ under random assignment.}

\label{algorithm}
\begin{algorithmic}[1]
\State{Maintain a set $T$ containing the heaviest $r$ elements seen so far at every moment (initialize $T$ with $r$ dummy elements that are supposed to be lighter than every element in the stream).}
\State{$\ALG \gets \emptyset$.}
\State{Choose $m$ from the binomial distribution $\Bin(n,p)$.}
\State{Observe the first $m$ elements and denote this set as the sample.}
\For{each element $e$ arriving after the first $m$ elements}
\If{$e$ enters $T$ and the element leaving $T$ is in the sample or is a dummy element}
\State{Check if $\ALG \cup \{e\}$ is independent. If so, add $e$ to $\ALG$.}
\EndIf
\EndFor
\State{Return the set $\ALG$.}
\end{algorithmic}
\end{algorithm}

\begin{thm}\label{thm:algorithm1}
Let $\ALG$ be the set returned by Algorithm \ref{algorithm} when applied to a uniformly dense matroid $\M$ of rank $r$. Then
\begin{equation*}
\E_{\sigma,\pi}[w(\ALG(\sigma))] \geq (-p^2\ln p)\sum_{i=1}^r w_i \geq (-p^2\ln p)\E_{\sigma}[w(\OPT_\M(\sigma))].
\end{equation*}
In particular, by setting $p=e^{-1/2}$, we conclude that Algorithm \ref{algorithm} is $2e$-competitive for uniformly dense matroids in the random assignment model.
\end{thm}
\begin{proof}

Consider the following offline simulation algorithm. In the first part of the simulation, each weight $w_i$ in the adversarial list $W$ selects an arrival time $t_i$ in $(0,1)$ uniformly and independently. The algorithm keeps a set $T$ containing the top $r$ \emph{weights} seen at every moment (initially containing $r$ dummy weights of negative value) and processes the weights as they arrive, sampling the ones arriving before time $p$. When a weight arriving after time $p$ is processed, the algorithm \emph{marks it as a candidate} if, first, that weight enters $T$, and second, the one leaving $T$ is either in the sample or a dummy weight.
In the second part of the simulation, the algorithm assigns to each candidate weight a different element of the matroid uniformly at random. Then, it runs the greedy procedure on the sequence of candidates in the order they appeared and returns its answer. Using that the cardinality of the sampled set has binomial distribution with parameters $n$ and $p$, it is not hard to check that the set of elements and weights returned by this simulation has the same distribution as the one returned by Algorithm \ref{algorithm}. For this reason, we focus on the simulation.

We estimate the probability that each one of the top $r$ weights appears in the output set. Focus on one such weight $w_i$, $i\leq r$, and let $\ell$ be a nonnegative integer strictly smaller than $r$. Define $\mathcal{E}_\ell$ as the event that exactly $\ell$ of the top $r$ weights excluding $w_i$ are sampled.  Given that event $\mathcal{E}_\ell$ occurs, the corresponding $\ell$ high weights in the sample enter $T$ as soon as they arrive and never leave it. Since every candidate pushes out either a dummy or a sampled weight of $T$ at the moment it is marked, the previous implies that the number of candidates marked by the simulation algorithm is at most (in fact, exactly) $r-\ell$. Event $\mathcal{E}_\ell$ occurs with probability $\binom{r-1}{\ell}p^\ell(1-p)^{r-1-\ell}$.

\begin{Claim} For $\ell \leq r-1, i \leq r$ and $t_i \geq p$,
$\Pr(w_i \text{ is marked as a candidate } | \mathcal{E}_\ell, t_i)\geq p/t_i$.
\end{Claim}

\begin{proof}
  Since $w_i$ is one of the top $r$ weights, it enters set $T$ at the time it is considered. Thus, it will be marked as a candidate if and only if the weight leaving $T$ at that time is either a dummy or a sampled weight.

Let $A(t_i) = \{w_j: t_j < t_i\}$ be the set of weights seen before $w_i$ arrives. If this set has less than $r$ elements then the element leaving $T$ at $t_i$ will be a dummy weight. Consider the case where $A(t_i)$ has cardinality at least $r$ and let $w_j$ be the top $r$-th element of this set. Since $w_j$ is not one of the top $r$ elements in the full adversarial list, its arrival time $t_j$ is independent of $\mathcal{E}_\ell$. Therefore,
\begin{align*}
  \Pr(w_i &\text{ is marked as a candidate } |\ \mathcal{E}_\ell, t_i) \\
   &=   1 \cdot \Pr(|A(t_i)|<r\ |\ \mathcal{E}_\ell, t_i) + \Pr(t_j < p\ |\ t_j < t_i) \Pr(|A(t_i)|\geq r\ |\ \mathcal{E}_\ell, t_i)\\
   &=   1 \cdot \Pr(|A(t_i)|<r\ |\ \mathcal{E}_\ell, t_i) + \frac{p}{t_i} \Pr(|A(t_i)|\geq r\ |\ \mathcal{E}_\ell, t_i) \geq \frac{p}{t_i}.\qedhere
\end{align*}
\end{proof}

The elements of the matroid assigned to the candidate weights form a random set. Conditioned on $\mathcal{E}_\ell$ and on $w_i$ being a candidate, Lemma \ref{lem:randomgreedy} implies that no matter what position $w_i$ takes in the list of at most $r-\ell$ candidates, the probability that it gets added to the output is at least $1-(r-\ell)/r = (\ell + 1)/r$; therefore, the probability that $w_i$ appears in the output is at least
\begin{align*}
\sum_{\ell=0}^{r-1} \frac{(\ell + 1)}{r} \binom{r-1}{\ell}p^\ell(1-p)^{r-1-\ell} \int_{p}^1\frac{p}{t_i}dt_i &= \frac{(rp + (1-p))(-p\ln p)}{r} \geq -p^2 \ln p.
\end{align*}
Theorem \ref{thm:algorithm1} follows easily from here.\qedhere
\end{proof}

We remark that Algorithm \ref{algorithm} does not need to learn the weights of the elements: the algorithm can proceed by only learning the relative order of the weights seen so far. Also, we note that this algorithm is not constant competitive in the zero information model. In fact, if we had such an algorithm $\mathcal{A}$ for uniformly dense matroids under zero information we could obtain one for general matroids by using that every matroid $\M$ is a restriction of a uniformly dense matroid $\M'$~\cite{lai1995every}. The algorithm for $\M$ would virtually complete the matroid $\M'$ by adding a dummy set of zero weight elements and run algorithm $\mathcal{A}$ on $\M'$, simulating the augmented input in such a way that the dummy elements arrive uniformly at random similarly to the real ones.

\section{Principal sequence and general matroids.}\label{section4}

In this section we  need the concept of \emph{principal sequence} of a matroid. This notion is related to the theory of \emph{principal partition} of graphs, matroids and submodular systems, which has applications to connectivity and reliability of networks and to resource allocation problems. The theory of principal partition is extensively analyzed in a monograph by Narayanan~\cite{narayanan1997submodular} and in a recent survey of Fujishige~\cite{fujishige2009theory}.
The definition of principal sequence of a matroid we present was introduced by several authors under different names (See, e.g., ~\cite{Narayanan74phd,tomizawa1976strongly,narayanan1981elementary}).

\begin{thm}[Principal Sequence]\label{thm:ppalsequence}
Let $\M=(E,\I)$ be a loopless matroid with rank function $r$. There is a sequence of sets $\emptyset = F_0 \subsetneq F_1 \subsetneq \dots \subsetneq F_k = E$ and a sequence of values $\infty  > \lambda_1 > \lambda_2 > \dots > \lambda_k \geq 1$ satisfying:
\begin{compactenum}
  \item The values $\lambda_1,\ldots,\lambda_k$ are the only ones for which the submodular set function $f_{\lambda_i}:=\lambda_i r(X) - |X|$ admits more than one minimizer.
  \item For every $1\leq i \leq k$, the unique minimal and unique maximal minimizers of the function $f_{\lambda_i}$ are $F_{i-1}$ and $F_i$ respectively.
\end{compactenum}
The sequence $(F_i)_{i=0}^k$ is called the \emph{principal sequence} of $\M$ and $(\lambda_i)_{i=1}^k$ is the associated sequence of critical values.
\end{thm}

From this definition, it is not hard to obtain the following lemma (it also follows from \cite[Theorem 3.11]{fujishige2009theory} or \cite{catlin1992fractional}).

\begin{lem}\label{lem:Density} Let $\M=(E,\I)$ be a loopless matroid with principal sequence $(F_i)_{i=0}^k$ and critical values $(\lambda_i)_{i=1}^k$. Then, for every $1\leq i\leq k$, the matroid $\M_i = \left(\M / F_{i-1}\right) \big|_{(F_i \setminus F_{i-1})}$ obtained by contracting $F_{i-1}$ and restricting to $F_i \setminus F_{i-1}$ is uniformly dense, with density $\frac{|F_i\setminus F_{i-1}|}{r_{\M_i}(F_i \setminus F_{i-1})} = \lambda_i$. These matroids are known as the \emph{principal minors} of~$\M$.
\end{lem}

The principal sequence of $\M$ can be constructed by iteratively finding the maximal densest set $F_i$ of the current matroid, adding it to the sequence and contracting it on the matroid, until all the elements have been contracted. Polynomial time algorithms to compute the principal sequence of a given matroid can be found in the literature (see, e.g. \cite{narayanan1981elementary} or \cite[chapters 10,11]{narayanan1997submodular}).

We use Lemma \ref{lem:Density} to design an algorithm for the matroid secretary problem under random assignment in a general (not necessarily loopless) matroid $\M=(E,\I)$. Let $(\M_i)_{i=1}^k$ be the sequence of principal minors of the loopless matroid obtained by deleting the set $E_0$ of loops from $\M$. For every $i$, let $E_i=F_i\setminus F_{i-1}$ and $r_i$ denote the ground set and rank of $\M_i$ respectively. Note that the family $\{E_0,\ldots,E_k\}$ is a partition of the ground set $E$. Define matroids $\M'$ and $\PP$ with ground set $E$ as follows:
\begin{align*}
\I(\M')= \left\{\bigcup_{i=1}^k I_i:\, I_i \in \I(\M_i)\right\}, \quad \I(\PP)= \left\{\bigcup_{i=1}^k I_i:\, I_i \subseteq E_i, |I_i| \leq r_i\right\}.
\end{align*}

Since any independent set in $\M'$ is, by definition of each $\M_i$, also independent in $\M$, Algorithm~\ref{algorithm2}, described below returns an independent set of $\M$.

\begin{algorithm}[h!!!!!]
\caption{for General Matroids of $n$ elements and rank $r$ under random assignment.}
\label{algorithm2}
\begin{algorithmic}[1]
\State{Compute the sequence of principal minors $(\M_i)_{i=1}^k$ of the matroid obtained by removing the loops of $\M$.}
\State{Run Algorithm \ref{algorithm} in parallel on each $\M_i$ and return the union of the answers.}
\end{algorithmic}
\end{algorithm}

\begin{thm}\label{thm:algorithm2}
Algorithm \ref{algorithm2} is $2e/(1-1/e)=2e^2/(e-1)$-competitive for general matroids in the random assignment model.
\end{thm}

To prove Theorem \ref{thm:algorithm2}, we compare the weight of the set $\ALG$ returned by Algorithm \ref{algorithm2} with the optimum of the partition matroid $\PP$ defined above. Since both $\M'$ and $\PP$ are disjoint union of uniformly dense and uniform matroids over the same ground set and having the same rank, we expect them to behave similarly.
Observe that the random permutation $\sigma:[n]\to [n]$ that is used to assign the weights of the adversary to the elements of the matroid can be viewed as the composition of a random partition of $[n]$ into blocks of size $(|E_i|)_{i=0}^k$, and a collection of random permutations inside each block. Conditioned on the random partition, each block $E_i$ receives a hidden list of weights which are assigned uniformly at random to the elements of the block. Since each $\M_i$ is uniformly dense and the elements of $E_i$ arrive in random order, Theorem \ref{thm:algorithm1} implies that Algorithm~\ref{algorithm} recovers, in expectation, at least $1/(2e)$-fraction of the \emph{combined weight of the heaviest $r_i$ elements of $E_i$}, where the expectation is over the random permutation of that particular block. Noting that the union of the heaviest $r_i$ elements of each $E_i$ is exactly the optimum of the partition matroid $\PP$ defined above, we conclude, by removing the conditioning, that
\begin{align}
  \E_{\sigma}[w(\OPT_\M(\sigma))] \geq   \E_{\sigma}[w(\OPT_{\M'}(\sigma))] \geq \E_{\sigma,\pi}[w(\ALG)] \geq \E_{\sigma}[w(\OPT_{\PP}(\sigma))]/(2e). \label{ineq}
\end{align}

To prove that Algorithm~\ref{algorithm2} is constant competitive we only need to show that the optimum of $\PP$ is only a constant away from the optimum in $\M$.

\begin{lem}\label{lem:ineq2} $\E_\sigma[w(\OPT_\PP(\sigma))] \geq (1-1/e)\E_\sigma[w(\OPT_\M(\sigma))]. $
\end{lem}

To prove this lemma we note the following fact. For all $j$, let $A^\sigma_j=\{e_{\sigma(1)},e_{\sigma(2)},\ldots,e_{\sigma(j)}\}$ denote the (random) set of the elements that receives the top $j$ weights. Then, for any matroid:
\begin{align}
  \E_\sigma[w(\OPT(\sigma))] &= \sum_{j=1}^n w_j \Pr_\sigma[e_{\sigma(j)} \not\in \spa(A^\sigma_{j-1})] =\sum_{j=1}^n w_j (\E_\sigma[r(A^\sigma_j)]-\E_\sigma[r(A^\sigma_{j-1})])\notag\\
  &= \E_\sigma[r(A^\sigma_n)] w_n + \sum_{j=1}^{n-1} \E_\sigma[r(A^\sigma_j)] (w_j - w_{j+1}).\label{useful}
\end{align}

In order to prove Lemma \ref{lem:ineq2}, we only need to show the following.

\begin{lem}\label{lem:expectedrank}
  For every $1\leq j\leq n$, $\E_\sigma[r_\PP(A^\sigma_j)] \geq (1-1/e)\E_\sigma[r_\M(A^\sigma_j)].$
\end{lem}
\begin{proof}
Since $P$ is a partition matroid with non-trivial parts $E_1,\ldots,E_k$, we have:
$$\E_\sigma[r_\PP(A^\sigma_j)] = \sum_{i=1}^k \E_\sigma[r_P(A^\sigma_j \cap E_i)] = \sum_{i=1}^k \E_\sigma[\min(|A^\sigma_j \cap E_i|,r_i)].$$

For each $i$, $\PP|_{E_i}$ is a uniform matroid with the same density $\lambda_i=|E_i|/r_i$ as the corresponding uniformly dense matroid $\M_i$. By Theorem \ref{thm:ppalsequence}, these densities are strictly decreasing with $i$.
Call a part $E_i$ \emph{dense} if $\lambda_i \geq n/j$ and \emph{sparse} otherwise. Intuitively, when a part $E_i$ is dense we expect $A^\sigma_j \cap E_i$ to contain at least $|E_i| (j/n) \geq r_i$ elements and thus we expect the rank of $A^\sigma_j \cap E_i$ in the partition matroid to be close to $r_i$. On the other hand, for sparse parts this quantity should be closer to $\E[|A^\sigma_j \cap E_i|] = |E_i| (j/n)$. We formalize this intuition in the following claim.

\begin{Claim} If $\lambda_i \geq n/j$, then $\E_\sigma[r_P(A^\sigma_j \cap E_i)] \geq (1-1/e) r_i$. If on the other hand, $\lambda_i \leq n/j$, then $\E_\sigma[r_P(A^\sigma_j \cap E_i)] \geq (1-1/e) |E_i| (j/n)$.
\end{Claim}

\begin{proof}
Focus on a part $E_i$ and split its elements into $r_i$ pieces as evenly as possible. To do this, let $\epsilon_i= \lambda_i - \lfloor \lambda_i \rfloor$ and create $r_i (1 - \epsilon_i)$ pieces of size $\lambda_i - \epsilon_i = \lfloor\lambda_i\rfloor$ and $r_i \epsilon_i$ pieces of size $\lambda_i + 1 -\epsilon_i=\lfloor\lambda_i\rfloor +1$. It is easy to see that both $r_i(1-\epsilon_i)$ and $r_i\epsilon$ are integers and that the previous construction is indeed a partition of $E_i$ into $r_i$ pieces as claimed (note that if $\lambda_i=|E_i|/r_i$ is an integer, this partition consists simply on $r_i$ pieces of size $\lambda_i$).

The rank of any set in $\PP|_{E_i}$ is at least as high as the number of pieces of $E_i$ this set intersects; therefore, $\E[r_{\PP}(A^\sigma_j \cap E_i)]$ is at least
\begin{align*}
\sum_{B: \text{Piece of $E_i$}} \Pr((A_j^\sigma \cap B)\neq 0) &= \sum_{B: \text{Piece of $E_i$}}\left(1-\frac{\binom{n-|B|}{j}}{\binom{n}{j}}\right) =\sum_{B: \text{Piece of $E_i$}}\left(1-\prod_{\ell=0}^{|B|-1}\left(1 - \frac{j}{n-\ell}\right)\right)\\
 &\geq\sum_{B: \text{Piece of $E_i$}} 1-(1 - j/n)^{|B|}\\
 &=r_i (1-\epsilon_i)\left(1-(1-j/n)^{\lambda_i-\epsilon_i}\right) + r_i \epsilon_i\left(1-\left(1-j/n\right)^{\lambda_i-\epsilon_i+1}\right)\\
  &=r_i\left(1 - (1-j/n)^{\lambda_i-\epsilon_i}\left(1-\epsilon_ij/n\right) \right)\\
   &\geq r_i\left(1 - e^{(-j/n)(\lambda_i - \epsilon_i)}e^{-\epsilon_i j/n}\right) = r_i(1-e^{-\lambda_i j/n}).
\end{align*}

The function $(1-e^{-x})$ is increasing, thus if $\lambda_i \geq n/j$,
$$\E[r_\PP(A^\sigma_j \cap E_i)] \geq r_i (1-e^{-\lambda_i j/n}) \geq r_i (1-e^{-1}).$$

On the other hand, the function $(1-e^{-x})/x$ is decreasing, thus if $\lambda_i \leq n/j$,
$$\E[r_\PP(A^\sigma_j \cap E_i)] \geq r_i (1-e^{-\lambda_i j/n}) = (r_i \lambda_i j/n) \frac{(1-e^{-\lambda_i j/n})}{\lambda_i j/n} \geq (|E_i| j/n) (1-e^{-1}).\qedhere$$  
\end{proof}
Since $(\lambda_i)_{i=1}^k$ is a decreasing sequence, there is an index~$i^*$ such that $E_i$ is dense if and only if $1\leq i\leq i^*$. Recall that $\bigcup_{i=1}^{i^*}E_i$ is equal to the set $F_{i^*}$ in the principal sequence of the matroid $\M\setminus E_0$. Since every set in the principal sequence has the same rank in both $\M$ and $\PP$ we get:
\begin{align*}
\E_\sigma[r_\M(A^\sigma_j)] &\leq \E_\sigma[r_\M(A^\sigma_j \cap F_{i^*}) + r_\M(A^\sigma_j \cap (E\setminus F_{i^*}))] \leq r_\M(F_{i^*}) + \E_\sigma[|A^\sigma_j \cap (E\setminus F_{i^*})| \\
&= \sum_{i=1}^{i^*} r_i + \sum_{i=i^*+1}^{k} |E_i| (j/n) \leq \sum_{i=1}^k\E_\sigma[r_P(A_j^\sigma \cap E_i)] / (1-1/e) = \E_\sigma[r_P(A_j^\sigma)]/(1-1/e).\qedhere
\end{align*}
\end{proof}

By combining Lemma \ref{lem:ineq2} with inequality \eqref{ineq} we conclude the proof of Theorem \ref{thm:algorithm2}.

\section{Algorithms for the zero information model}\label{section5}
We give various algorithms for different classes of matroids in the zero information model. In all of them we assume the matroid is \emph{loopless} (in the random assignment model we can not make this assumption since the introduction of loops changes the distribution of the weights of the elements).

\subsection{Cographic Matroids}
In any $3$-edge-connected graph $G$ we can find three spanning trees $T_1$, $T_2$ and $T_3$, such that the union of their complements covers $E(G)$ (This follows from e.g. Edmonds's Matroid Partitioning Theorem~\cite{edmonds1965minimum}). The sets $B_i=E\setminus T_i$ are bases in the cographic matroid of $G$. Consider the following algorithm for this matroid. Select $i \in \{1,2,3\}$ uniformly at random and accept all elements in $B_i$. Since every edge of $G$ is selected with probability at least $1/3$, this algorithm is $3$-competitive.

We modify this algorithm to work on the cographic matroid $\M$ of any graph $G$. First, remove all the bridges of $G$ since they are loops in $\M$. Decompose the edge set of the remaining graph as the direct sum of $2$-edge-connected components.  For each component $C$, let $C'$ be the graph obtained by contracting all but one edge in each serial class of its corresponding graphic matroid.\footnote{Two elements are in series in a graphic matroid if and only if they are in parallel in the cographic matroid. A pair of elements $\{e,f\}$ are in parallel in a matroid, if the set $\{e,f\}$ is a circuit. Being in series and being in parallel are equivalence relations, so the serial and parallel class of an element are well defined.
Contracting all but one edge in each serial class of the graphic matroid corresponds to deleting all but one element in each parallel class of the cographic matroid. For the specific case of the cographic matroid of a graph, a set of elements are \emph{in parallel} if each pair of them is a minimal edge cut of the graph.}
Each graph $C'$ is then $3$-edge-connected and, as before, we can find three bases $B_1$, $B_2$ and $B_3$ of the cographic matroid of $C'$ covering $E(C')$. The algorithm for $\M$ is as follows. Independently for each component $C$, select an index $i \in \{1,2,3\}$ uniformly at random and run the $e$-competitive algorithm of~\cite{babaioff_knapsack_2007} on the partition matroid that accepts at most one edge of $C$ from each parallel class of $\M$ represented in $B_i$ (discard every element of $C$ not represented in $B_i$). Since every element of the optimum base of $\M$ is the heaviest of its parallel class and each parallel class of $\M$ is selected with probability at least $1/3$, we conclude the previous algorithm is $3e$-competitive.

\begin{thm} For any cographic matroid $\M$, the previous algorithms is $3e$-competitive. Furthermore, if the graph $G$ associated to $\M$ is $3$-edge-connected, the algorithm is $3$-competitive.
\end{thm}
\subsection{Low Density Matroids}
A generalization of the previous algorithm is the following. Given a loopless matroid $\M=(E,\I)$ of density $\gamma(\M)$, the vector $\vec{v} \in \R^E$ having all its coordinates equal to $1/\gamma(\M)$ is feasible in the matroid polytope. In particular, this vector has a decomposition as convex combination of independent sets of $\M$: $\vec{v} = \sum_{I \in \I}\lambda_{I} \chi_{I}$, which we can find in polynomial time. The algorithm for matroid $\M$ will select an independent set $I \in  \I$ at random, according to probabilities $(\lambda_{I})_{I \in \I}$ and accept its elements without looking at their weights. Since every element $e$ is selected with probability $\sum_{I \in \I: I \ni e}\lambda_I = \vec{v}_e = 1/\gamma(\M)$, this algorithm is $\gamma(\M)$-competitive.

If matroid $\M$ contains parallel elements, we could get a better competitive ratio by considering the simple matroid $\M'=(E',\I')$ obtained by removing all but one edge in each parallel class of $\M$. By combining the output $I'$ of the previous algorithm applied on $\M'$ with the $e$-competitive algorithm for the partition matroid that selects one element in each parallel class represented in $I'$ (similar to what we did for cographic matroids), we obtain a $e\gamma(\M')$-competitive algorithm.

\begin{thm} For any matroid $\M$, the first algorithm described is $\gamma(\M)$-competitive, and the second algorithm is $\gamma(\M')e$-competitive.
\end{thm}

\subsection{Matroids with small cocircuits}
For each element $e$ of a loopless matroid $\M=(E,\I)$, let $c^*(e)$ be the size of the smallest cocircuit (i.e. circuits of the dual matroid) containing it, and let $c^*(\M)=\max_{e} c^*(e)$. Consider the algorithm that greedily construct an independent set of $\M$ selecting elements as they appear without looking at their weights. We claim this algorithm is $c^*(\M)$-competitive. To see this, fix an element $e \in E$ and let $C^*$ be a cocircuit of minimum size containing it. If $e$ appears before all the other elements of $C^*$ in the random order then it has to be selected by the algorithm. Otherwise, there would be a circuit $C$ that intersects $C^*$ only in element $e$, which is a contradiction (See, e.g. \cite[Proposition 2.1.11]{oxley2006matroid}). It is not hard to prove that for every matroid $\M$, $\gamma(\M)\leq c^*(\M)$. This mean that this algorithm is no better than the one for low density matroids. However, this algorithm is much simpler.

\begin{thm} For any matroid $\M$, the algorithm described above is $c^*(\M)$-competitive.
\end{thm}
\subsection{Column-sparse linear matroids}
Let $\M=(V,\I)$ be a linear matroid represented by a matrix $A$ containing at most $k$ non-zero values in each column.
Consider the following algorithm: Randomly permute the rows of $A$ and define for every row $i$, the sets $C_i =\{v \in V: v_i \neq 0\}$ and $B_i =C_i \setminus \bigcup_{j<i} C_j$, where $v_i$ denotes the $i$-th coordinate of column $v$ in the permuted matrix. Next, run the secretary algorithm for the partition matroid that accepts at most one element of each $B_i$. We claim that any set returned by this algorithm is independent in $\M$: If this was not the case there would be a circuit $C$ inside the output. Let $v \in C$ be the element belonging to the set $B_i$ of smallest index~$i$. By definition of $v$, the elements of $C\setminus v$ are not in $C_i$; therefore, $C$ and $C_i$ intersects only in $v$. This is a contradiction since $C_i$ is in the cocircuit space of the matroid (Use, e.g. \cite[Proposition 2.1.11]{oxley2006matroid}).

We claim the previous algorithm is $ke$-competitive. To see this, construct the bipartite graph $G$ with parts the rows and columns of $A$, where there is an edge $(j,v)$ if the corresponding entry of $A$ is non-zero. Assign to each edge a weight equal to the one of its associated column in $\M$. Consider the following simulation algorithm: Randomly permute the vertices in the row part of the graph. Delete all the edges, except the ones going from a column vertex to its lowest neighbor (the row having smallest index in the random permutation). Finally, run the secretary algorithm for the partition matroid that accepts for each row vertex, at most one edge incident to it. This returns a matching with the same weight as the set of elements the original algorithm returns.

Since for every independent set of columns, the number of row vertices that this set dominates in $G$ is at least its cardinality, Hall's Theorem implies that there is a matching covering each independent set. In particular the weight of the maximum weight matching $M^*$ in $G$ is at least the one of the optimum independent set of $\M$. On the other hand, $M^*$ has weight at most the one of the edge set $\{(i,v^*(i)): i \in \text{rows}(A)\}$, where $v^*(i)$ is the maximum weight neighbor of $i$ in $G$. Since each edge $(i,v^*(i))$ is not deleted with probability $1/k$  and, given it is not deleted, the simulation selects it with probability $1/e$, we conclude the original algorithm is $ke$-competitive.

\begin{thm} The previous algorithm is $ke$-competitive for matroids representable by matrices having only $k$ non-zero elements per column.
\end{thm}

Note that by applying this algorithm to graphic matroids, which are representable by matrices having only 2 ones per column, we recover the $2e$ competitive algorithm of Korula and Pal \cite{korula2009algorithms}.

\subsection{A new $O(\log r)$ competitive algorithm for matroids}

Babaioff et al.~\cite{babaioff_matroids_2007} present an $O(\log r)$ competitive algorithm for general matroids of rank $r$. This algorithm has many features, including the fact that it does not need to know the matroid beforehand; it only needs to know the number of elements and have access to an oracle that test independence only on subsets of elements it has already seen. Nevertheless, this algorithm makes use of the actual values of the weights being revealed. We present an algorithm having the same features but that only uses the relative order of weights seen and not their numerical value.

If we are given the rank of the matroid, our algorithm is as follows. With probability 1/2, run the classical secretary algorithm that returns the heaviest element of the stream. Otherwise, observe the first $m$ elements of the stream, where $m$ is chosen from the binomial distribution $\Bin(n,1/2)$ (as usual, denote this set of elements as the \emph{sample}) and compute the optimum base $A=\{a_1,\dots,a_k\}$ (with $w(a_1)\geq \dots \geq w(a_k)$) of the sampled elements. Afterwards, select a number $\ell \in \{1,3,9,\dots,3^t\}$ with $t=\lfloor \log_3 r\rfloor$ uniformly at random, run the greedy procedure on the set of non-sampled elements having weight at least the one of $a_{\ell}$ as they arrive and return its answer (if $\ell > k$, run the greedy procedure over the entire set of non-sampled elements). If we are not given the rank of the matroid beforehand, we select $t \in \{\lfloor \log_3 k\rfloor, \lfloor \log_3 k\rfloor +1\}$ uniformly and use this value in the previous algorithm.

The optimum of the sample is similar to the optimum of the nonsampled part: For any number $\ell$ the algorithm can choose, there is an independent set of size close to $\ell$ outside the sample with every element heavier than $a_\ell$ (with high probability); therefore, the greedy procedure recovers a weight of roughly $\ell w(a_\ell)$. By taking the expectation over the choices of $\ell$ it is not hard to check that the expected weight returned by the algorithm is at least $\Omega(\E[w(A)/ \log_3(r)])=\Omega(\E[w(\OPT)/ \log_3(r)])$. We give the formal proof below.

\begin{thm}\label{thm:log}
  The algorithm described above is $O(\log r)$-competitive for any matroid of rank $r$.
\end{thm}
\begin{proof}
Assume first that the rank $r$ of the matroid is known. Let $\OPT=\{e_1,\ldots,e_r\}$ with $w(e_1)\geq \dots \geq w(e_r)$ be the maximum independent set of the matroid and $A=\{a_1,\dots,a_k\}$ be the optimum of the sample, i.e., the optimum of the first $\Bin(n,1/2)$ elements in the stream (independent of whether the algorithm computes $A$ or not). Note that every element of the matroid is sampled independently with probability $1/2$, including the elements of the optimum. Therefore,

\begin{equation}
\E[w(A)]\geq \frac{w(\OPT)}{2}. \label{eqn:OPT2}
\end{equation}

To simplify our analysis, in the following we assume that for $i>k$, $a_i$ is a dummy element with $w(a_i)=0$. Given the number $\ell$ chosen by the algorithm (if the algorithm reaches that state), the weight of the set returned will be at least $w(a_\ell)$ times the number of elements the greedy procedure selects; therefore, $\E[w(\ALG)]$ is at least
\begin{align*}
\frac{w(e_1)}{2e} + \frac{1}{2(1+\lfloor \log_3r\rfloor)}\sum_{j = 0}^{\lfloor \log_3r\rfloor}\E\left[w(a_\ell)\cdot |\ALG|\, \big|\ \text{$\ell = 3^j$ was selected}\right].
\end{align*}

Let $H(a_\ell)$ be the collection of non-sampled elements that are heavier than $a_\ell$. If the algorithm chooses the number $\ell$, it will then execute the greedy procedure on $H(a_\ell)$ and return a set of cardinality equal to the rank of $H(a_\ell)$. Note that for every~$\ell$, $w(e_\ell)\geq w(a_\ell)$; therefore, the rank of $H(a_\ell)$ is at least the number of nonsampled elements in $\{e_1,\ldots,e_\ell\}$. By Chernoff bound, the probability that this last quantity is smaller than $\ell/4$ is at most $\exp(-\ell/8)$.

In particular, if $\ell\geq 9$,  $\E[w(a_{\ell})\cdot |\ALG|\, \big|\ \ell] \geq \E[w(a_\ell)](1-\exp(-\ell/8)) \ell/4\geq \E[w(a_\ell)]\ell/6$. Therefore,
\begin{align*}
\E[w(\ALG)] &\geq \frac{w(e_1)}{2e} + \frac{1}{12(1+\lfloor \log_3r\rfloor)}\sum_{j = 2}^{\lfloor \log_3r\rfloor} \E[w(a_{3^j})]3^j\\
&\geq \E\left[ \frac{w(a_1,\dots,a_8)}{16e} + \frac{1}{24(1+\lfloor \log_3r\rfloor)}\sum_{j = 2}^{\lfloor \log_3r\rfloor}w(a_{3^j},\dots,a_{3^{j+1}-1}) \right]\\
&\geq \frac{\E[w(A)]}{16e(1+\lfloor \log_3 r\rfloor)}.
\end{align*}
Using inequality \eqref{eqn:OPT2}, we get
\begin{align*}
\E[w(\ALG)] &\geq \frac{w(\OPT)}{32e(1+\lfloor \log_3 r\rfloor)},
\end{align*}
which implies the algorithm is $O(\log r)$-competitive.

Suppose now that the rank $r$ is unknown. If $r$ is small, say $r \leq 12$, then with probability $1/(2e)$ the algorithm will run the standard secretary algorithm and return the top element of the matroid. This element has weight at least $1/12$ fraction of the optimum; therefore the algorithm is $24e$-competitive for this case.

For the case where $r>12$ we use a different analysis.  The random variable $k$ denoting the rank of the sampled set could be strictly smaller than $r$. However, the probability that $k\leq r/3$ is small. Indeed, for that event to happen we require that at most $1/3$ of the elements of $\OPT$ are in the sample. By Chernoff bound, this happens with probability $\exp(-r/18)\leq \exp(-13/18)\leq 1/2$. Noting that $r/3\leq k \leq r$ implies that $\lfloor \log_3 r \rfloor \in \{\lfloor \log_3 k \rfloor, \lfloor \log_3 k \rfloor +1\}$, we deduce that with probability at least 1/4 our algorithm guesses $t=\lfloor \log_3 r \rfloor $ right; therefore, the competitive ratio of this algorithm is at most 4 times worse than the one that knows the rank beforehand.
\end{proof}
\bibliographystyle{abbrv}

\bibliography{matsec}

\end{document}